\documentclass[12pt,cls,onecolumn]{IEEEtran}
\usepackage{graphicx,amsmath,amssymb,epsfig, amsfonts, cite, latexsym, cuted, multicol, multirow, subfigure, stfloats, array, tabularx}
\usepackage{subeqnarray}
\usepackage{color}
\usepackage{setspace}
\usepackage{anysize}

\begin{document}

\title{Elastic Routing in Wireless Ad Hoc Networks With Directional Antennas}
\author{\large Jangho Yoon, \emph{Member}, \emph{IEEE}, Won-Yong Shin, \emph{Member}, \emph{IEEE}, \\ and Sang-Woon Jeon, \emph{Member}, \emph{IEEE}
\\
\thanks{This research was supported by the Basic Science Research
Program through the National Research Foundation of Korea (NRF)
funded by the Ministry of Education (2014R1A1A2054577). The
material in this paper was presented in part at the IEEE
International Symposium on Information Theory, Honolulu, HI,
June/July 2014~\cite{YoonShinJeon:13}.}
\thanks{J. Yoon is with
the Attached Institute of ETRI, Daejeon 305-390, Republic of
Korea. (E-mail: yunjh@kaist.ac.kr).}
\thanks{W.-Y. Shin (corresponding author) is with the Department of Computer Science and
Engineering, Dankook University, Yongin 448-701, Republic of Korea
(E-mail: wyshin@dankook.ac.kr).}
\thanks{S.-W. Jeon is
with the Department of Information and Communication Engineering,
Andong National University, Andong 760-749, Republic of Korea
(E-mail: swjeon@anu.ac.kr).}
} \maketitle




\newtheorem{definition}{Definition}
\newtheorem{theorem}{Theorem}
\newtheorem{lemma}{Lemma}
\newtheorem{example}{Example}
\newtheorem{corollary}{Corollary}
\newtheorem{proposition}{Proposition}
\newtheorem{conjecture}{Conjecture}
\newtheorem{remark}{Remark}

\def \diag{\operatornamewithlimits{diag}}
\def \min{\operatornamewithlimits{min}}
\def \max{\operatornamewithlimits{max}}
\def \log{\operatorname{log}}
\def \max{\operatorname{max}}
\def \rank{\operatorname{rank}}
\def \out{\operatorname{out}}
\def \exp{\operatorname{exp}}
\def \arg{\operatorname{arg}}
\def \E{\operatorname{E}}
\def \tr{\operatorname{tr}}
\def \SNR{\operatorname{SNR}}
\def \dB{\operatorname{dB}}
\def \ln{\operatorname{ln}}

\def \bmat{ \begin{bmatrix} }
\def \emat{ \end{bmatrix} }

\def \be {\begin{eqnarray}}
\def \ee {\end{eqnarray}}
\def \ben {\begin{eqnarray*}}
\def \een {\end{eqnarray*}}

\begin{abstract}
Throughput scaling laws of an {\em ad hoc} network equipping {\em
directional} antennas at each node are analyzed. More
specifically, this paper considers a general framework in which
the beam width of each node can scale at an arbitrary rate
relative to the number of nodes. We introduce an {\em elastic
routing} protocol, which enables to increase per-hop distance
elastically according to the beam width, while maintaining an
average signal-to-interference-and-noise ratio at each receiver as
a constant. We then identify fundamental operating regimes
characterized according to the beam width scaling and analyze
throughput scaling laws for each of the regimes. The elastic
routing is shown to achieve a much better throughput scaling law
than that of the conventional nearest-neighbor multihop for all
operating regimes. The gain comes from the fact that more
source--destination pairs can be simultaneously activated as the
beam width becomes narrower, which eventually leads to a linear
throughput scaling law. In addition, our framework is applied to a
hybrid network consisting of both wireless ad hoc nodes and
infrastructure nodes. As a result, in the hybrid network, we
analyze a further improved throughput scaling law and identify the
operating regime where the use of directional antennas is
beneficial.
\end{abstract}


\newpage

\section{Introduction}



As the number of devices explosively increases for
machine-to-machine (M2M) communications in the era of the internet
of things (IoT), characterizing the aggregate throughput of {\em
large-scale} wireless networks becomes more crucial in developing
transmission protocols efficiently delivering a number of packets.
While numerical results via computer simulations depend heavily on
specific operating parameters for a given system or protocol, a
study on the capacity scaling of large-scale networks with respect
to the number of nodes provides a fundamental limit on the network
throughput. Hence, one can obtain remarkable insights into the
practical design of a protocol by characterizing the capacity
scaling law.

\subsection{Related Work}

In~\cite{GuptaKumar:00}, throughput scaling was originally
introduced and characterized in a large-scale wireless \emph{ad
hoc} network. It was shown that, for a network having $n$ nodes
randomly distributed in an unit area (i.e., a dense network), the
aggregate throughput scales as $\Omega(\sqrt{n/\log n})$ by
conveying packets in the nearest-neighbor multihop routing
fashion.\footnote{We use the following notation: i) $f(x)=O(g(x))$
means that there exist constants $C$ and $c$ such that $f(x)\leq
Cg(x)$ for all $x>c$, ii) $f(x)=o(g(x))$ means that
$\lim_{x\rightarrow \infty}\frac{f(x)}{g(x)}=0$, iii)
$f(x)=\Omega(g(x))$ if $g(x)=O(f(x))$, and iv) $f(x)=\Theta(g(x))$
if $f(x)=O(g(x))$ and $g(x)=O(f(x))$.} There have been further
studies based on multihop routing in the
literature~\cite{GuptaKumar:03,DousseFranceschettiThiran:06,XueXieKumar:05,ElGamalMammenPrabhakarShah:06,ElGamalMammen:06,FranceschettiDouseTseThiran:07,ShinChungLee:13},
while the total throughput scales far less than $\Theta(n)$.
In~\cite{OzgurLevequeTse:07}, the aggregate throughput in the
dense network was improved to an almost linear scaling, i.e.,
$\Theta(n^{1-\epsilon})$ for an arbitrarily small $\epsilon>0$, by
using a hierarchical cooperation strategy. Besides the
hierarchical cooperation
scheme~\cite{OzgurLevequeTse:07,NiesenGuptaShah:09,NiesenGuptaShah:10,JeonGastpar:14},
there has also been a steady push to improve the throughput of
interference-limited networks up to a linear scaling by using node
mobility~\cite{GrossglauserTse:02,ElGamalMammenPrabhakarShah:06},
interference alignment~\cite{CadambeJafar:08}, infrastructure
support~\cite{KozatTassiulas:03,ZemlianovVeciana:05,LiuLiuTowsley:03,LiuThiranTowsley:07,ShinJeonDevroyeVuChungLeeTarokh:08,LiFang:10,JeonChung:12},
and directional antennas~\cite{YiPeiKalyanaraman:03,ZhangLiew:06}.

To achieve such a linear scaling, there will be a price to pay in
terms of
delay~\cite{OzgurLevequeTse:07,GrossglauserTse:02,ElGamalMammenPrabhakarShah:06},
cost of channel
estimation~\cite{OzgurLevequeTse:07,CadambeJafar:08}, and
infrastructure
investment~\cite{ZemlianovVeciana:05,ShinJeonDevroyeVuChungLeeTarokh:08}.
On the one hand, the use of {\em directional}
antennas~\cite{YiPeiKalyanaraman:03,ZhangLiew:06,Ramanathan_etal:05}
in ad hoc networks has recently emerged as a promising technology
leading to the enhanced spatial reuse, the improved transmission
distance, and the reduced interference level with relatively low
cost in comparison to alternative technologies. Especially, for
wireless systems using millimeter wave (mmWave) technologies
operating in the 10-300 GHz band, which have been considered as
one solution to enable gigabit-per-second data rates, equipping
directional antennas at each node may be more challenging. This is
because mmWave links are inherently directional and thus steerable
antenna arrays can be easily implemented, thus resulting in a much
higher link gain \cite{ThornbugBaiHeath:14}. Due to these reasons,
the interest in
 studies of more amenable networks using directional antennas has been greatly growing.
%
In the literature, the previous work based on the {\em protocol
model}~\cite{GuptaKumar:00} has shown that, for an infinitely
large antenna gain (or equivalently, for an infinitely small beam
width), the use of directional antennas provides a substantial
throughput enhancement up to a linear
scaling~\cite{LiZhangFang:11}. Similarly, the throughput scaling
law was studied based on an {\em interference model} for
directional antennas~\cite{ZhangLiew:06}. There have also been
other research directions showing the capacity scaling when
directional antennas are used under different assumptions as well
as different
situations~\cite{PerakiServetto:03,ZhangXuWangGuizani:10}. It is,
however, still unclear how much throughput scaling gain is
attainable by directional antennas under a more realistic wireless
channel. More precisely, for {\em all beam width scaling}
conditions, the analysis so far would not be suitable for fully
understanding the effects of directional antennas on the
throughput scaling since it does not reflect the physically
attainable antenna gain that may be exploited to enhance the
throughput scaling in the presence of interference.

\subsection{Contributions}

In this paper, when there are $n$ randomly located nodes in a
network equipping directional antennas at each node, we deal with
a general framework in which the beam width of each node,
$\theta$, scales at an {\em arbitrary rate} with respect to $n$,
which provides a comprehensive understanding on fundamental limits
of directional antennas for wireless ad hoc networks. Similarly as
in~\cite{YiPeiKalyanaraman:03,LiZhangFang:11}, we take into
account a simplified but feasible directional antenna model having
mainlobe and sidelobe gains. Then, by completely utilizing the
characteristics of directional antennas, we introduce a new
routing, termed {\em elastic routing}, and analyze its throughput
scaling laws. The proposed routing protocol enables to increase
the average transmission distance at each hop {\em elastically}
according to the beam width, while setting the average
signal-to-interference-and-noise ratio (SINR) at each receiver to
a constant independent of $n$. We identify two fundamental {\em
operating regimes} characterized according to the $\theta$ scaling
and analyze throughput scaling laws for each of the regimes. Our
main results demonstrate that the proposed elastic routing
achieves a much better throughput scaling law compared to the
conventional nearest-neighbor multihop routing for all operating
regimes. The gain comes from the fact that more
source--destination (SD) pairs can be activated simultaneously as
the beam width $\theta$ becomes narrower, which eventually
provides up to a linear throughput scaling.
Interestingly, it is further shown that the average delay is
reduced by the proposed elastic routing while achieving an
improved throughput scaling, which is in sharp contrast with the
omnidirectional mode (i.e.,
$\theta=\Theta(1)$)~\cite{ElGamalMammenPrabhakarShah:06}.

In addition, our result is extended to a hybrid network scenario.
Since there will be a long latency and insufficient energy with
only wireless connectivity, it would be good to deploy
infrastructure nodes, or equivalently base stations (BSs) in the
network, while possibly improving the throughput scaling. In a
hybrid network equpping directional antennas at each node, we
analyze the impact and benefits of our elastic routing in further
improving the throughput scaling.

%
%
%

 Our comprehensive analysis sheds more light on the
routing policies and on operational issues for ad hoc networks in
the directional mode.

\subsection{Organization}

The rest of this paper is organized as follows. The system model
is described in Section \ref{SEC:System}. In Section
\ref{SEC:routing}, the elastic routing protocol is presented. In
Section \ref{SEC:Throughput}, the throughput scaling laws are
derived. In Section \ref{SEC:DISCUSS}, our result is extended to a
hybrid network model. Finally, Section \ref{SEC:CONCL} summarizes
the paper with some concluding remarks.


\section{System Model} \label{SEC:System}

We consider a two-dimensional wireless ad hoc network consisting
of $n$ nodes that are uniformly distributed at random on a square.
The nodes are grouped into $n/2$ SD pairs at random. Each node
operates in half-duplex mode and is equipped with a single
directional antenna. The network area is assumed to be {one and
$n$ in dense and extended} networks, respectively.


\begin{figure}[t!]
  \begin{center}
  \leavevmode \epsfxsize=0.35\textwidth   
  \leavevmode 
\epsffile{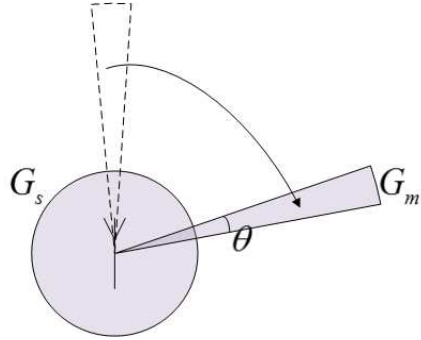} \caption{The beam pattern of
the directional antenna model under consideration.}
\label{FIG:BeamPattern}
  \end{center}
\end{figure}

We define a hybrid antenna model whose mainlobe is characterized
as a sector and whose sidelobe forms a circle (backlobes are
ignored in this model).\footnote{Instead of directional antennas
modeled as a cone in a three-dimensional
view~\cite{LiZhangFang:11}, we use a rather simple two-dimensional
antenna model since simplifying the shape of the antenna pattern
will not cause any fundamental change in terms of capacity scaling
law.} As illustrated in Fig.~\ref{FIG:BeamPattern}, the antenna
beam pattern has a gain value $G_m$ for the mainlobe of beam width
$\theta\in [0,2\pi)$, and also has a sidelobe of gain $G_s$ of
beam width $2\pi-\theta$. The parameters $G_m$ and $G_s$ are then
related according to
\begin{align}
\frac{\theta}{2\pi}G_m+\frac{2\pi-\theta}{2\pi}G_s = 1, \nonumber
\end{align}
where $0\le G_s\le 1 \le G_m$. In our work, we assume that
$G_m=\Theta(1/\theta)$ and $G_s=\Theta(1)$, which do not violate
the law of conservation of energy.\footnote{Note that $G_m=G_s=1$
in the omnidirectional mode.} For simplicity, we assume unit
antenna efficiency, i.e., no antenna loss. Each node can steer its
antenna for directional transmission or directional reception.


Similar to~\cite{LiZhangFang:11}, for a given time instance,
suppose that node $i\in\{1,\cdots, n\}$ transmits to node
$k\in\{1,\cdots, n\}\setminus \{i\}$ and they beamform to each
other according to the above assumption. Let $\mathcal{I}_1$,
$\mathcal{I}_2$, and $\mathcal{I}_3$ denote three different sets
of nodes transmitting at the same time, where both nodes $i_1 \in
\mathcal{I}_1$ and $k$ beamform to each other, either node $i_2
\in \mathcal{I}_2$ or $k$ beamforms to the other node (but not
both), and neither node $i_3 \in \mathcal{I}_3$ nor $k$ beamforms
to the other node, respectively. Note that node $i$ is in
$\mathcal{I}_1$ since nodes $i$ and $k$ beamform to each other.

Under the directional antenna model, the received signal of node
$k$ at the given time instance is represented by
\begin{align}
y_k = \sum_{i_1\in \mathcal{I}_1}h_{ki_1}x_{i_1} + \sum_{i_2\in
\mathcal{I}_2}h_{ki_2}x_{i_2} + \sum_{i_3\in
\mathcal{I}_3}h_{ki_3}x_{i_3} + n_k, \nonumber
\end{align}
where $x_{i_1}$, $x_{i_2}$, and $x_{i_3}\in\mathbb{C}$ are the
signals transmitted by nodes $i_1$, $i_2$, and $i_3$,
respectively, and $n_k$ denotes the circularly symmetric complex
Gaussian noise with zero mean and variance $N_0$. The channel
coefficients $h_{ki_1}$, $h_{ki_2}$, and $h_{ki_3}$ are given by
\begin{subequations}
\begin{align}
h_{ki_1}&=\frac{G_m e^{j\phi_{ki_1}}}{r_{ki_1}^{\alpha/2}}, \nonumber\\
h_{ki_2}&=\frac{\sqrt{G_mG_s}
e^{j\phi_{ki_2}}}{r_{ki_2}^{\alpha/2}},
\nonumber\\
h_{ki_3}&=\frac{G_s e^{j\phi_{ki_3}}}{r_{ki_3}^{\alpha/2}},
\nonumber
\end{align}
\end{subequations}
respectively, where $\phi_{kj}$ represents the random phase
uniformly distributed over $[0,2\pi)$ and independent for
different $j$, $k$, and time (transmission symbol), i.e., fast
fading~\cite{OzgurLevequeTse:07,ShinJeonDevroyeVuChungLeeTarokh:08,OzgurJohariTseLeveque:10}.
The parameters $r_{kj}$ and $\alpha>2$ denote the distance between
nodes $j$ and $k$, and the path-loss exponent, respectively. Each
node should satisfy the average transmit power constraint $P>0$
during transmission, which is a constant. Channel state
information (CSI) is assumed to be available at all receivers, but
not at transmitters.

In the following, we formally define the per-node and aggregate
throughputs used throughout the paper.

\begin{definition}[Throughput] \label{def:throughput}
A per-node throughput $R(n)$ is said to be \emph{achievable} with
high probability (w.h.p.) if all sources can transmit at the rate
of $R(n)$ bits/sec/Hz to their destinations with probability
approaching one as $n$ increases. Accordingly, the achievable
aggregate throughput is at least given by $T(n)=(n/2)R(n)$.
\hfill$\lozenge$
\end{definition}

For the rest of this paper, we will analyze throughput scaling
laws of wireless ad hoc networks equipping directional antennas at
each node. Throughout the paper, $\mathbb{E}[\cdot]$ and
$\Pr(\cdot)$ denote the statistical expectation and the
probability, respectively. Unless otherwise stated, all logarithms
are assumed to be to base two.

\section{Elastic Routing Protocol} \label{SEC:routing}
%
%

In this section, we describe our elastic routing protocol, which
can be designed with the help of the directional antennas. Under
the protocol, we perform multihop (or even single-hop)
transmission by \emph{elastically} increasing per-hop distance as
a function of the scaling parameter $\theta$, which ultimately
enhances the throughput performance compared to the conventional
multihop~\cite{GuptaKumar:00,ElGamalMammenPrabhakarShah:06}. We
focus primarily on the dense network configuration, but the
overall procedure of our protocol can be directly applied to the
extended network configuration.

\begin{figure}[t!]
  \begin{center}
  \leavevmode \epsfxsize=0.43\textwidth   
  \leavevmode 
\epsffile{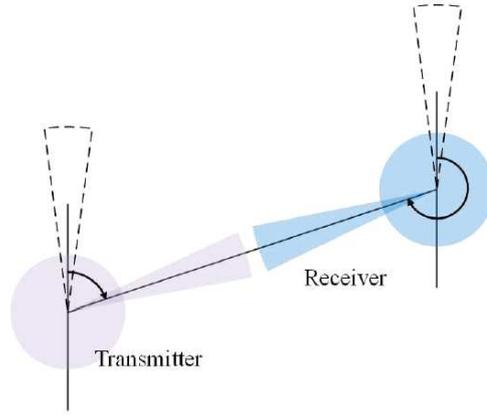} \caption{Beam steering at each
transmitter--receiver pair.} \label{FIG:TX_RX_BF}
  \end{center}
\end{figure}

Let us first show how to perform beam steering at each node using
directional antennas. At each hop, as illustrated in Fig.
\ref{FIG:TX_RX_BF}, the antennas of each selected
transmitter--receiver pair are steered so that their beams cover
each other, which enables to achieve the maximum antenna gain $G_m^2$ at each transmitter--receiver pair. 

\begin{figure*}[t!]
  \begin{center}
  \leavevmode \epsfxsize=0.71\textwidth   
  \leavevmode 
\epsffile{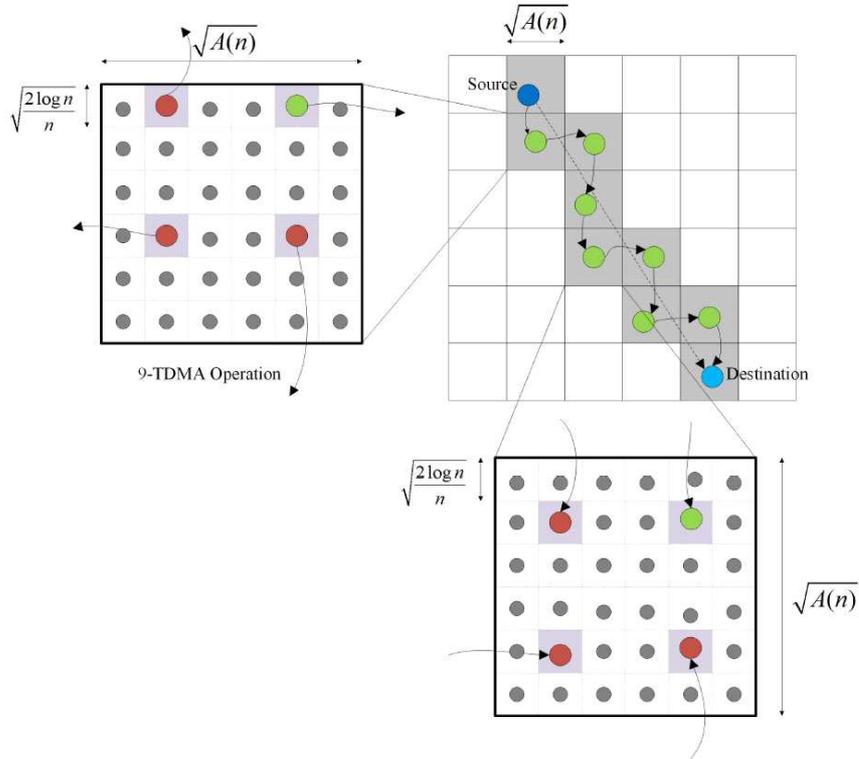} \caption{An example of an SD
path passing through its associated routing regions in the dense
network, where the 9-TDMA scheme is used.} \label{FIG:SD_Pair}
  \end{center}
\end{figure*}

We now describe how each SD pair performs the elastic routing
based on the beam steering technique described above. Let $T_S$ be
the total number of scheduling time slots. We assume that, at each
time slot $s \in \{1,\cdots,T_S\}$, randomly chosen $M(n)$ SD
pairs are activated simultaneously, where  $M(n)$ scales as
$\Omega(\sqrt{n/\log n})$ and $O(n)$.\footnote{It is not desirable
that $M(n)$ scales as $o(\sqrt{n/\log n})$ since, in this case,
the throughput scaling achieved by our elastic routing is less
than that of the conventional multihop scheme.} We denote the set
of scheduling time slots to which the $p$th SD pair belongs by
$\Phi_p$, where $p\in\{1,\cdots,n/2\}$. For instance, if the $p$th
SD pair is scheduled at time slots $1$, $3$, and $5$, $\Phi_p$ is
given by $\Phi_p = \{1, 3, 5\}$. We further denote $T_p \triangleq
|\Phi_p |$, where $|\Phi_p |$ denotes the cardinality of $\Phi_p$.

The following lemma shows that each of $n/2$ SD pairs can be
served with almost the same fraction of time in the limit of large
$n$.

\begin{lemma}[Strong typicality] \label{LMM:typicality}
Suppose that $T_S=n^4$. Then, for sufficiently large $n$,
$\frac{T_p}{T_S}$ is lower-bounded by
\begin{align}
\frac{2M(n)}{n} - \frac{1}{n} \nonumber
\end{align}
w.h.p. for all $p \in \{1,\cdots,n/2 \}$.
\end{lemma}

\begin{proof}
Since $M(n)$ SD pairs among $n/2$ are scheduled at random for each
time slot, the probability that a specific SD pair is served at
each time slot is given by $2M(n)/n$. Hence, by applying the
result of \cite[Lemma 2.12]{CsiszarKorner:81}, we show that, for
any $\xi > 0$, the probability that
\begin{align}
\left|\frac{T_p}{T_S} - \frac{2M(n)}{n} \right| \le \xi
\nonumber 
\end{align}
for all $p \in \{1,\cdots,n/2\}$ is greater than $1-\frac{n}{4 T_S
\xi^2}$. Then, by setting $\xi = \frac{1}{n}$, we have
$\frac{T_p}{T_S} \ge \frac{2M(n)}{n} - \frac{1}{n}$ with
probability greater than $1-\frac{1}{4n}$, which converges to one
as $n$ tends to infinity. This completes the proof of the lemma.
\end{proof}

As depicted in Fig.~\ref{FIG:SD_Pair}, we divide the whole area
into $1/A(n)$ square {\em routing cells} with per-cell area
$A(n)$, where $A(n)$ is assumed to scale as $\Omega(\log n/n)$ and
$O(1)$. We draw the straight line connecting a source to its
destination, termed an SD line. Then, packets for each SD pair
travel horizontally or vertically along its SD line by hopping
along adjacent routing cells of area $A(n)$ until they reach the
corresponding destination. While travelling along its SD line, a
certain node in each routing cell is arbitrarily selected as a
relay forwarding the packets. As in the earlier
work~~\cite{GuptaKumar:00,ElGamalMammenPrabhakarShah:06}, in the
dense network, a transmit power of $P (\log n / n)^{\alpha / 2}$
is used at each node.

We further divide each routing cell with area $A(n)$ into smaller
square regions of area $2\log n / n$, which guarantees that each
smaller region has at least one node w.h.p.
(see~\cite{GuptaKumar:00} for the details).
 As
in~\cite{GuptaKumar:00,ElGamalMammenPrabhakarShah:06}, the 9-time
division multiple access (TDMA) scheme is used between smaller
regions to avoid huge interference (see also
Fig.~\ref{FIG:SD_Pair}). Then, only a single node in each
activaved smaller region transmits its packets.

Figure~\ref{FIG:SD_Pair} illustrates the packet transmission of a
particular SD pair, where the source and destination nodes are
indicated by blue circles, the relay nodes are indicated by green
circles, and the interfering nodes that simultaneously transmit
packets of other SD pairs are indicated by red circles,
respectively.

Let $d_{\text{hop}}$ and $\bar{h}$ denote the average transmission
distance at each hop and the average number of hops per SD pair,
respectively. Then, from the above routing, it follows that
\begin{align}
A(n)=\Theta\left(d_{\text{hop}}^2\right)=\Theta\left(\frac{1}{\bar{h}^2}\right)
\label{EQ:An_d_hop}
\end{align}
in the dense network.

Now, let us turn to how to decide per-cell area $A(n)$ according
to given beam width $\theta$, which plays an important role in
determining the throughput of the proposed elastic routing. If
$A(n)$ is given by a function of $n$, then the average per-hop
distance $d_{\text{hop}}$ can also be determined using
(\ref{EQ:An_d_hop}). We note that $d_{\text{hop}}$ is elastically
increased as much as possible while the average SINR at each
receiver is set to $\Theta(1)$.
Such $d_{\text{hop}}$ is shown according to the value of $\theta$
as follows:
\begin{align} \label{EQ:d_hop}
d_{\text{hop}}&=\Theta\left(\min\left\{\sqrt{\frac{\log n}{n}}
\theta^{-\frac{2}{\alpha}} , 1 \right\}  \right) \nonumber\\
&= \left\{
\begin{array}{ll}
 \Theta\left(\sqrt{\frac{\log n}{n}}
\theta^{-\frac{2}{\alpha}}\right) & \textrm{if
$\theta^{-1}=o\left(\left(\frac{n}{\log
n}\right)^{\alpha/4}\right)$}\\
 \Theta\left(1\right) &\textrm{if
$\theta^{-1}=\Omega\left(\left(\frac{n}{\log n} \right)^{\alpha/4}
\right)$},
\end{array}
\right.
\end{align}
which will be verified later. Note that
$d_{\text{hop}}=\Theta(\sqrt{\log n / n})$ in the omnidirectional
mode (i.e., $\theta=\Theta(1)$).

As expressed in (\ref{EQ:d_hop}), according to the beam width
$\theta$, the whole operating regimes are divided into two
fundamental regimes. The two operating regimes and their
corresponding routing schemes in each regime are summarized as
follows.

\begin{itemize}
%
%
\item {\bf Regime I}: $\theta^{-1} = o \left( \left( \frac{n}{\log
n}\right)^{\alpha/4} \right)$.

In the regime, the elastic routing outperforms the
nearest-neighbor multihop routing. As $\theta^{-1}$ increases, a
long-range transmission is performed at each hop owing to an
enhanced antenna gain at each transmitter--receiver pair.

\item {\bf Regime II}: $\theta^{-1} = \Omega \left( \left(
\frac{n}{\log n}\right)^{\alpha/4} \right)$.

The single-hop transmission is performed in the regime, where
per-hop distance can reach up to $O\left(1\right)$.
\end{itemize}

In the extended network, a full transmit power $P$ is used at each
node. Then, it follows that
\begin{align}
d_{\text{hop}}&=\Theta\left( \min  \left\{ \sqrt{\log
n}\theta^{-2/\alpha} ,\sqrt{n} \right\} \right) \nonumber\\ &=
\left\{
\begin{array}{ll}
 \Theta\left(\sqrt{\log
n}\theta^{-2/\alpha}\right) & \textrm{if
$\theta^{-1}=o\left(\left(\frac{n}{\log
n}\right)^{\alpha/4}\right)$}\\
 \Theta\left(\sqrt{n}\right) &\textrm{if
$\theta^{-1}=\Omega\left(\left(\frac{n}{\log n} \right)^{\alpha/4}
\right)$},
\end{array}
\right. \label{EQ:d_hop_extended}
\end{align}
which turns out to be scaled up by a factor of $\sqrt{n}$,
compared to the dense network case in~(\ref{EQ:d_hop}).

Since, unlike the nearest-neighbor multihop
routing~\cite{GuptaKumar:00}, packets can travel much farther at
each hop for the proposed elastic routing in the directional mode,
the average number of hops, $\bar{h}$, can be reduced
significantly, thereby resulting in an improved throughput up to a
linear scaling, which will also be specified in the following
section.

\begin{figure*}[t!]
  \begin{center}
  \leavevmode \epsfxsize=0.94\textwidth   
  \leavevmode 
\epsffile{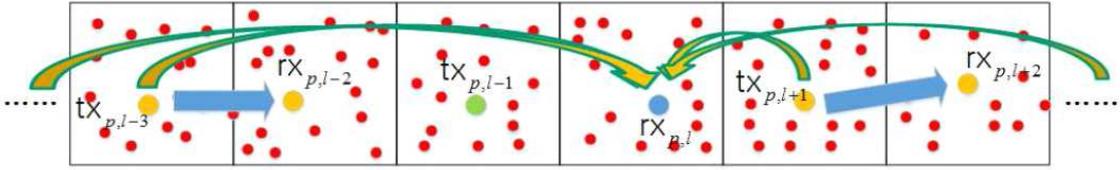} \caption{The
intra-pair interference affecting ${\sf rx}_{p,l}$ with correlated
beam directions, where ${\sf tx}_{p,l}$ is the desired
transmitter.} \label{FIG:I_MODEL_INTER_HOP}
  \end{center}
\end{figure*}


\section{Main Results} \label{SEC:Throughput}

This section presents our main result, which shows throughput
scaling laws achievable by elastic routing for both dense and
extended networks using directional antennas.



\subsection{Throughput Scaling in Dense Networks} \label{SEC:Throughput_Dense}

In this subsection, the throughput scaling for dense networks
under the elastic routing protocol is analyzed. Let ${\sf
tx}_{p,l}$ and ${\sf rx}_{p,l}$ denote the transmitter and the
corresponding receiver of the $l$th hop of the {$p$}th SD pair,
respectively, where {$l\in\mathcal{H}_p$} and
{$p\in\{1,\cdots,n/2\}$}. Recall that
{$\mathcal{H}_p=\{1,2,\cdots,\beta_p \bar{h}\}$} denotes the set
of hops for the {$p$}th SD pair, where {$\beta_p>0$} is a
parameter that scales as $O(1)$. Here, ${\sf tx}_{p,l}$ and ${\sf
rx}_{p,l}$ are fixed during the entire scheduling time since they
are solely determined by the network geometry. Let
$\text{SINR}_{p,l}(s)$ denote the instantaneous received SINR of
${\sf rx}_{p,l}$ at time slot $s \in \{1,\cdots, T_S\}$ for the
$l$th hop of the {$p$}th SD pair. Then, we have
\begin{align}
\text{SINR}_{p,l} (s)& =\frac{P_{p,l}(s)}{N_0+I_{p,l}(s)},
\label{EQ:SINR}
\end{align}
where $P_{p,l}(s)$ denotes the received signal power of  ${\sf
rx}_{p,l}$ from the desired transmitter ${\sf tx}_{p,l}$ at time
slot $s$ and {$I_{p,l}(s)$} denotes the total interference power
of ${\sf rx}_{p,l}$ from all interfering nodes at time slot $s$.

Obviously, one can see that
\begin{align}
P_{p,l}(s)& = 0 \mbox{ if } s\notin \Phi_p \label{EQ:Pr1}
\end{align}
from the proposed elastic routing. For $s\in \Phi_p$, the $p$th SD
pair is activated and nodes ${\sf tx}_{p,l}$ and ${\sf rx}_{p,l}$
direct their beams to each other to maximize the antenna gain
during the $l$th hop transmission. Hence, we have
\begin{align}
P_{p,l}(s)& = |h_{{\sf rx}_{p,l}{\sf tx}_{p,l}}|^2
P\left(\frac{\log
n}{n} \right)^{\alpha / 2} \nonumber \\
& =\Omega\left(G_m^2 \left(\frac{\log n}{d_{\text{hop}}^2 n}
\right)^{\alpha / 2}\right) \mbox{ if } s\in \Phi_p.
\label{EQ:Pr2}
\end{align}

Now, we turn to computing the total interference power
$I_{p,l}(s)$. For analytical convenience, we first divide
$I_{p,l}(s)$ into two parts, $I^{[1]}_{p,l}(s)$ and
$I^{[2]}_{p,l}(s)$, which indicate the intra-pair interference
power and the inter-pair interference power at time slot
$s\in\{1,\cdots, T_S\}$, respectively. Due to our multihop-based
elastic routing characteristics, the beam directions at the
transmitters and receivers belonging to the same SD pair may be
highly correlated to each other.
Figure \ref{FIG:I_MODEL_INTER_HOP} illustrates how the intra-pair
interference is generated along with correlated beam directions
when multiple transmitters in each SD pair are simultaneously
activated. For this reason, we treat the intra-pair interference
separately from the inter-pair interference. Then, it follows
that\footnote{Due to the 9-TDMA scheme, even if the distance
between ${\sf rx}_{p,l}$ and each intra-pair interferer needs to
be more carefully considered, it does not fundamentally change the
scaling law result for $I^{[1]}_{p,l}(s)$.}
\begin{align}
I^{[1]}_{p,l}(s) & \le 2P\left(\frac{\log n}{n}
\right)^{\alpha/2}\Biggl(\sum_{l=1}^{\beta_p \bar{h}}
\frac{G_m^2}{\left(l \cdot d_{\text{hop}} \right)^{\alpha}}\Biggr)
\nonumber \\
&= O\left(G_m^2 \left(\frac{\log n}{d_{\text{hop}}^2 n}
\right)^{\alpha / 2} \right), \label{EQ:P_IHI}
\end{align}
where the inequality holds since at most two nodes can generate
the intra-pair interference between $(l-1)d_{\text{hop}}$ and $l
d_{\text{hop}}$ apart from node ${\sf rx}_{p,l}$. Each receiver
also suffers from the inter-pair interference, which is generated
by other activated SD pairs. Unlike the intra-pair interference
case, the routing path of all SD pairs is determined independently
of each other since $M(n)$ SD pairs are chosen uniformly at random
at each time slot $s$. As illustrated in Fig. \ref{FIG:I_MODEL},
when the 9-TDMA scheme is used, the nodes generating the
inter-pair interference can be computed by using the layering
technique~\cite{OzgurLevequeTse:07,ShinJeonDevroyeVuChungLeeTarokh:08}
with the concept of \emph{tier} indexed by $t$, where ${\sf
in}_{p,l,t,i}(s)$ denotes the $i$th inter-pair interferer placed
on the $t$th tier that causes the inter-pair interference to ${\sf
rx}_{p,l}$ at time slot $s$. Note that ${\sf in}_{p,l,t,i}(s)$ may
vary over time slots since $M(n)$ SD pairs are uniformly chosen at
random for each time slot $s\in\{1,\cdots,T_S\}$.
 Then,
the total amount of inter-pair interference of  ${\sf rx}_{p,l}$
at time $s$ is upper-bounded by
\begin{align}
I^{[2]}_{p,l}(s)& \leq \sum_{t=1}^{\infty} \sum_{i=1}^{8t}
I^{[2]}_{p,l,t,i}(s)  \nonumber \\
& \le P \left( \frac{\log n}{n} \right)^{\alpha/2} \nonumber \\
&~~~ \cdot\sum_{t=1}^{\infty} \sum_{i=1}^{8t} \left(t
\sqrt{\frac{\log n}{n}}\right)^{-\alpha} X_{p,l,t,i}(s)
\nonumber \\
& = P\sum_{t=1}^{\infty} \sum_{i=1}^{8t}t^{-\alpha}
X_{p,l,t,i}(s), \label{EQ:Inter_Pair_Interference}
\end{align}
where $I^{[2]}_{p,l,t,i}(s)$ denotes the inter-pair interference
power of ${\sf rx}_{p,l}$ caused by the inter-pair interferer
${\sf in}_{p,l,t,i}(s)$ at time slot $s$ and  $X_{p,l,t,i}(s) \in
\{G_m^2, G_m G_s, G_s^2\}$ denotes the antenna gain between  ${\sf
rx}_{p,l}$ and ${\sf in}_{p,l,t,i}(s)$ at time slot $s$.


Before presenting our main result, we start from the following
lemma, which establishes an upper bound on the expected inter-pair
interference power at each receiver.

\begin{figure}[t!]
  \begin{center}
  \leavevmode \epsfxsize=0.7\textwidth   
  \leavevmode 
\epsffile{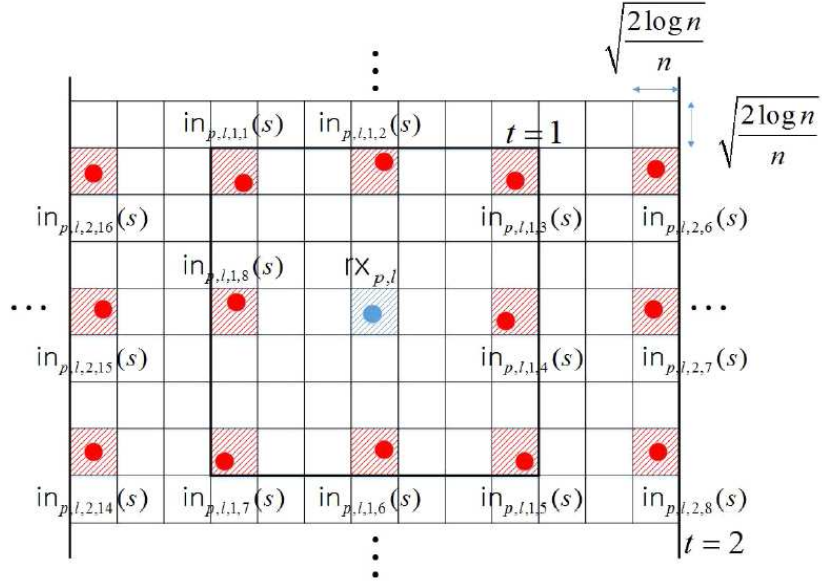} \caption{The
inter-pair interferers that affect ${\sf rx}_{p,l}$ when the
9-TDMA is used in the dense network.} \label{FIG:I_MODEL}
  \end{center}
\end{figure}

\begin{lemma} \label{LEM:Int1}
Consider the {$p$}th SD pair whose packets travel over {$\beta_p
\bar{h}$} hops from the source to its destination. Then,
the expectation of the inter-pair interference power
$I^{[2]}_{p,l}(s)$ is upper-bounded by
\begin{align}
\mathbb{E}\left[I^{[2]}_{p,l}(s) \right] = O(1). \label{EQ:avgPIu}
\end{align}
\end{lemma}

\begin{proof}
First of all, similarly as
in~\cite{YiPeiKalyanaraman:03,LiZhangFang:11}, the direction of
the receive beam  at node ${\sf rx}_{p,l} $ is independent of the
direction of the transmit beam at ${\sf in}_{p,l,t,i}$ for
$t\in\{1,2,\cdots\}$ and $i \in \{1,\cdots, 8t\}$ since inter-pair
interferers are only concerned.
Thus, from (\ref{EQ:Inter_Pair_Interference}), the expectation of
the inter-pair interference power is upper-bounded by
\begin{align} \label{EQ:GAIN_PR2}
&\mathbb{E}\left[I^{[2]}_{p,l}(s) \right] \nonumber\\
&\le \mathbb{E}\left[P\sum_{t=1}^{\infty}
\sum_{i=1}^{8t}t^{-\alpha}
X_{p,l,t,i}(s) \right]\nonumber \\
& = P\mathbb{E} \left[X_{p,l,t,i}(s) \right] \sum_{t=1}^{\infty}
\sum_{i=1}^{8t}t^{-\alpha} \nonumber\\
& = 8P\mathbb{E} \left[X_{p,l,t,i}(s)
\right]\sum_{t=1}^{\infty}t^{1-\alpha},
\end{align}
where the first equality holds since $\mathbb{E}
\left[X_{p,l,t,i}(s) \right]$ is the same for all $t \in
\{1,2,\cdots,\}$ and $i \in \{1,\cdots,8t \}$. By using the fact
that
\begin{subequations} \label{EQ:GAIN_PR}
\begin{align}
\Pr\left(X_{p,l,t,i}(s)= G_m^2 \right) &= \frac{\theta^2}{4 \pi^2},  \\
\Pr\left(X_{p,l,t,i}(s)= G_m G_s \right) &= \frac{(2\pi - \theta)\theta}{2 \pi^2} , \\
\Pr\left(X_{p,l,t,i}(s)= G_s^2 \right)&= \frac{(2\pi -
\theta)^2}{4 \pi^2},
\end{align} \label{EQ:GAIN_PR}
\end{subequations}
we have
\begin{align} \label{EQ:GAIN_PR3}
&\mathbb{E} \left[ X_{I2}^{k_{(p,l)}(t,i)}(s)\right]\nonumber\\
&=\frac{\theta^2}{4 \pi^2} G_m^2 + \frac{(2\pi - \theta)\theta}{2
\pi^2} G_m G_s + \frac{(2\pi - \theta)^2}{4 \pi^2} G_s^2
\nonumber\\ & =\Theta(1),
\end{align}
which comes from the fact that $G_m=\Theta(1/\theta)$ and
$G_s=\Theta(1)$. Finally, from \eqref{EQ:GAIN_PR2} and
\eqref{EQ:GAIN_PR3} and the fact that
$\sum_{t=1}^{\infty}t^{1-\alpha}$ is upper-bounded by a constant
for $\alpha>2$, we have (\ref{EQ:avgPIu}), which completes the
proof of the lemma.
\end{proof}

In the following theorem, we state the aggregate throughput
achievable by elasting routing in the dense network.

\begin{theorem} \label{THM:ELASTIC}
In the dense ad hoc network of unit area, the aggregate throughput
achieved by elastic routing is given by
\begin{align} \label{EQ:Tn_dense}
\lefteqn{T(n) = \Omega
\left(\min\left\{\sqrt{n}\theta^{-2/\alpha},n\right\}n^{-\epsilon}\right)}& \nonumber\\
&=\left\{
\begin{array}{lll}
\Omega \left(n^{1/2-\epsilon} \theta^{-2/\alpha}\right)
&\textrm{if
$\theta^{-1}\!\!=\!\!o\left(\left(\frac{n}{\log n}\right)^{\alpha/4}\right)$} \\
\Omega(n^{1-\epsilon}) & \textrm{if
$\theta^{-1}\!\!=\!\!\Omega\left(\left(\frac{n}{\log
n}\right)^{\alpha/4}\right)$}
\end{array}
\right.
\end{align}
w.h.p., where $\epsilon>0$ is an arbitrarily small constant.
\end{theorem}


\begin{proof}
First of all, we set $T_S=n^4$, which satisfies the condition in
Lemma \ref{LMM:typicality}. Since the per-node throughput $R(n)$
in Definition \ref{def:throughput} is defined as the average rate
per each SD pair over the entire time slots,

\begin{align} \label{EQ:SINR11}
R(n)&\geq \frac{1}{T_S}\min_{p,l}\left\{\sum_{s\in\{1,\cdots,T_S\}}\log\left(1+\text{SINR}_{p,l}(s)\right)\right\}\nonumber\\
&=
\frac{1}{T_S}\min_{p,l}\left\{\sum_{s\in\Phi_p}\log\left(1+\text{SINR}_{p,l}(s)|s\in\Phi_p\right)\right\}
\end{align}
is achievable, where the minimum is taken over all pairs
$p\in\{1,\cdots,n/2\}$ and hops $l\in\mathcal{H}_p$. Here, the
equality holds since $\text{SINR}_{p,l}(s)=0$ if $s\notin\Phi_p$
(see \eqref{EQ:SINR} and \eqref{EQ:Pr1}).

Then, from (\ref{EQ:SINR}),
\eqref{EQ:Pr2}--\eqref{EQ:Inter_Pair_Interference},
$\text{SINR}_{p,l}(s)$ can be lower-bounded by
\begin{align}
\text{SINR}_{\min} (s)\triangleq \frac{c_0G_m^2 \left(\frac{\log
n}{d_{\text{hop}}^2 n} \right)^{\alpha / 2}}{N_0+c_1 G_m^2
\left(\frac{\log n}{d_{\text{hop}}^2 n} \right)^{\alpha /
2}\!\!\!\!\!+I^{[2]}_{\max}(s)},
\end{align}
which does not rely on parameters $p$ and $l$, if $s\in\Phi_p$ for
all $p\in\{1,\cdots,n/2\}$ and $l\in\mathcal{H}_p$, where
\begin{align}
I^{[2]}_{\max}(s)\triangleq P\sum_{t=1}^{\infty}
\sum_{i=1}^{8t}t^{-\alpha} X_{p,l,t,i}(s). \nonumber
\end{align}
Hence, one can obtain
\begin{align}
&R(n)\nonumber\\
&\geq  \frac{1}{T_S}\sum_{s\in\Phi_p}\log\left(1+\text{SINR}_{\min} (s)|s\in\Phi_p\right)\nonumber\\
&\!\!\!\!\!\overset{\mbox{\footnotesize w.h.p.}}{\geq}\frac{T_p}{T_S}\mathbb{E}\left[\log\left(1+\text{SINR}_{\min} (s)|s\in\Phi_p\right)\right]\nonumber\\
&\geq \frac{T_p}{T_S}\log\left(1+\frac{c_0G_m^2 \left(\frac{\log
n}{d_{\text{hop}}^2 n} \right)^{\alpha / 2}}{N_0+c_1 G_m^2
\left(\frac{\log n}{d_{\text{hop}}^2 n} \right)^{\alpha /
2}\!\!\!\!\!+\mathbb{E}\left[I^{[2]}_{\max}(s)\right]}\right)\nonumber\\
&\!\!\!\!\!\overset{\mbox{\footnotesize
w.h.p.}}{\geq}\frac{T_p}{T_S}\log\left(1+\frac{c_0G_m^2
\left(\frac{\log n}{d_{\text{hop}}^2 n} \right)^{\alpha /
2}}{N_0+c_1 G_m^2 \left(\frac{\log n}{d_{\text{hop}}^2 n}
\right)^{\alpha /
2}\!\!\!\!\!+c_2}\right)\nonumber\\
&\!\!\!\!\!\overset{\mbox{\footnotesize
w.h.p.}}{\geq}\left(\frac{2M(n)}{n}-\frac{1}{n}\right)\log\left(1+\frac{c_0G_m^2
\left(\frac{\log n}{d_{\text{hop}}^2 n} \right)^{\alpha /
2}}{N_0+c_1 G_m^2 \left(\frac{\log n}{d_{\text{hop}}^2 n}
\right)^{\alpha / 2}\!\!\!\!\!+c_2}\right) \label{EQ:R_n_lower}
\end{align}
is achievable w.h.p., where $d_{\text{hop}}$ is given by
(\ref{EQ:d_hop}) and $c_0$, $c_1$, and $c_2$ are some positive
constants. Here, the first inequality holds since
$\text{SINR}_{\min} (s)$ is the same for all
$p\in\{1,\cdots,n/2\}$ and $l\in\mathcal{H}_p$, the second
inequality holds since $I^{[2]}_{\max}(s)$ is independent and
identically distributed (i.i.d.) for $s$ and $T_p$ in an
increasing function of $n$ (refer to Lemma \ref{LMM:typicality}),
the third inequality holds by Jensen's inequality since
$\log(1+a/x)$ is convex in $x$ for all $a>0$, the fourth
inequality holds from Lemma \ref{LEM:Int1}, and the fifth
inequality holds from Lemma \ref{LMM:typicality}.

Since
\begin{align}  \label{EQ:avgTn}
G_m^2 \left(\frac{\log n}{d_{\text{hop}}^2 n} \right)^{\alpha / 2}
= \left\{ \begin{array}{ll} \Theta(1) & \mbox{if
$\theta^{-1}\!\!=\!\!o\left(\left(\frac{n}{\log
n}\right)^{\alpha/4}\right)$}\\\Omega(1) & \mbox{if
$\theta^{-1}\!\!=\!\!\Omega\left(\left(\frac{n}{\log
n}\right)^{\alpha/4}\right)$} \end{array} \right.
\end{align}
from (\ref{EQ:d_hop}), by substituting  \eqref{EQ:avgTn} into
\eqref{EQ:R_n_lower}, it is shown that
$R(n)=\Omega\left(\frac{M(n)}{n}\right)$ is achievable w.h.p. and,
as a consequence,  $T(n)= \Omega\left(M(n)\right)$ is achievable
w.h.p. Finally, from the fact that
\begin{align}
M(n)=\Theta\left(\frac{n}{\bar{h}\log
n}\right)=\Theta\left(\frac{d_{\text{hop}}n}{\log n}\right)
\nonumber
\end{align}
and (\ref{EQ:d_hop}), the aggregate throughput achievable by
elastic routing is given by (\ref{EQ:Tn_dense}) w.h.p. This
completes the proof of the theorem.
\end{proof}

\begin{figure}[t!]
  \begin{center}
  \leavevmode \epsfxsize=0.46\textwidth   
  \leavevmode 
\epsffile{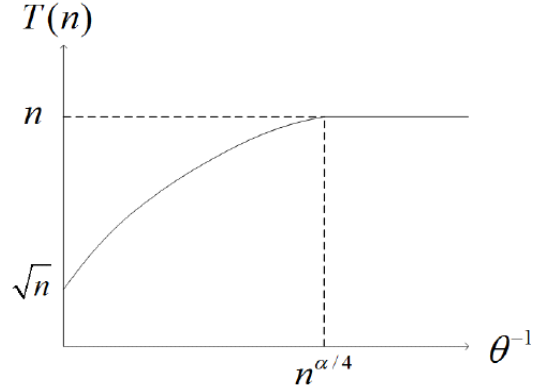} \caption{The aggregate throughput
scaling $T(n)$, achieved by elastic routing, with respect to the
inverse of the beam width, $1/\theta$.} \label{FIG:Throughput}
  \end{center}
\end{figure}


The aggregate throughput scaling $T(n)$ is illustrated in
Fig.~\ref{FIG:Throughput} according to the scaling parameter
$1/\theta$ (the terms $\epsilon$ and $\log n$ are omitted for
notational convenience). From Theorem~\ref{THM:ELASTIC}, each
operating regime is now closely scrutinized.
When $\theta^{-1} = o\big( ( n/\log n)^{\alpha/4} \big)$ (Regime
I), the throughput scaling gets increased as the beam width of the
directional antenna becomes narrower. This is the regime where the
proposed elastic routing provides a significant throughput
improvement over the nearest-neighbor multihop transmission with
increasing antenna gain. If $\theta^{-1}=\Omega \left(n/\log
n\right)^{\alpha/4}$ (Regime II), then a linear throughput scaling
is achieved, which corresponds to the fundamental limit of the
network under consideration and is consistent with the previous
work (including the beam width scaling condition) based on the
protocol model~\cite{LiZhangFang:11}. This is because all SD pairs
can be simultaneously activated with no degradation on the
received SINR, which scales as $\Omega(1)$, while maintaining
per-node throughput as a constant.
Our result is thus general in the sense that the achievable scheme
and its throughput are shown for all operating regimes with
respect to $\theta$ (i.e., for an arbitrary scaling of $\theta$).


%
%

From the achievability result, the following two interesting
discussions are also shown.

\begin{remark}[Delay analysis] \label{Rem:delay}
An improved aggregate throughput leads to a delay reduction owing
to elastic routing, whereas, by using the conventional multihop
scheme in the omnidirectional mode, the delay increases
proportionally with throughput
$T(n)$~\cite{ElGamalMammenPrabhakarShah:06}. \hfill$\lozenge$
\end{remark}

\begin{remark}[Ideal antenna model]
For comparison, let us consider an ideal antenna model. It can be
straightforwardly shown that Lemma~\ref{LEM:Int1} also holds when
there is no sidelobe gain~\cite{ZhangXuWangGuizani:10}, i.e.,
$G_s=0$. This reveals that, under this ideal model, the aggregate
throughput $T(n)$ is given by (\ref{EQ:Tn_dense}) as long as
$G_s=O(1)$. Therefore, the existence of sidelobe beams, whose gain
scales as $O(1)$, does not cause any throughput loss in scaling
law. \hfill$\lozenge$
\end{remark}

\begin{figure}[t!]
\centering{%
\subfigure[Omnidirectional antenna]{%
\epsfxsize=0.43\textwidth \leavevmode \label{FIG:delay_elgamal}
\epsffile{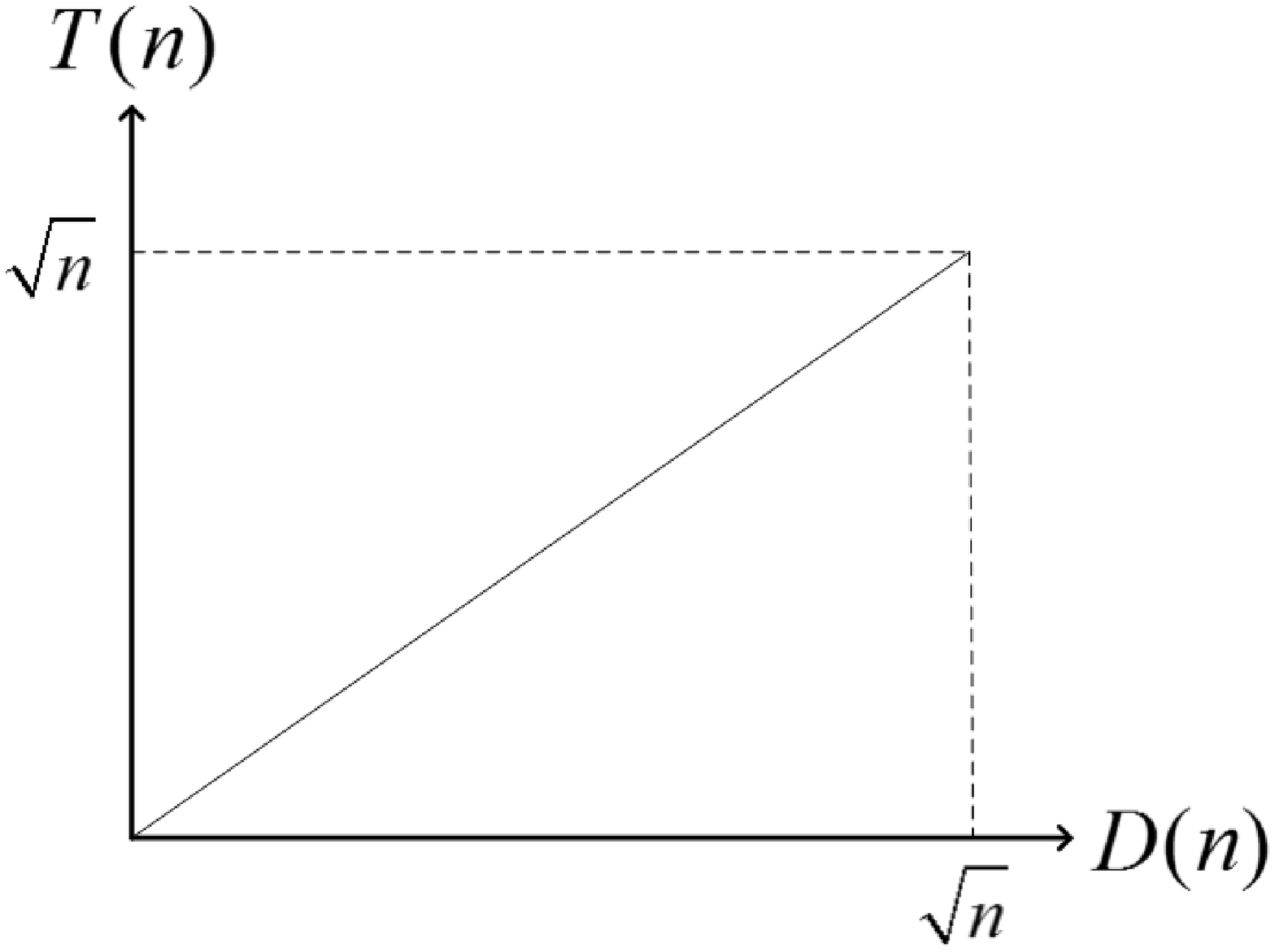}}\vspace{.0cm}
\subfigure[Directional antenna]{%
\epsfxsize=0.43\textwidth \leavevmode \label{FIG:delay_elastic}
\epsffile{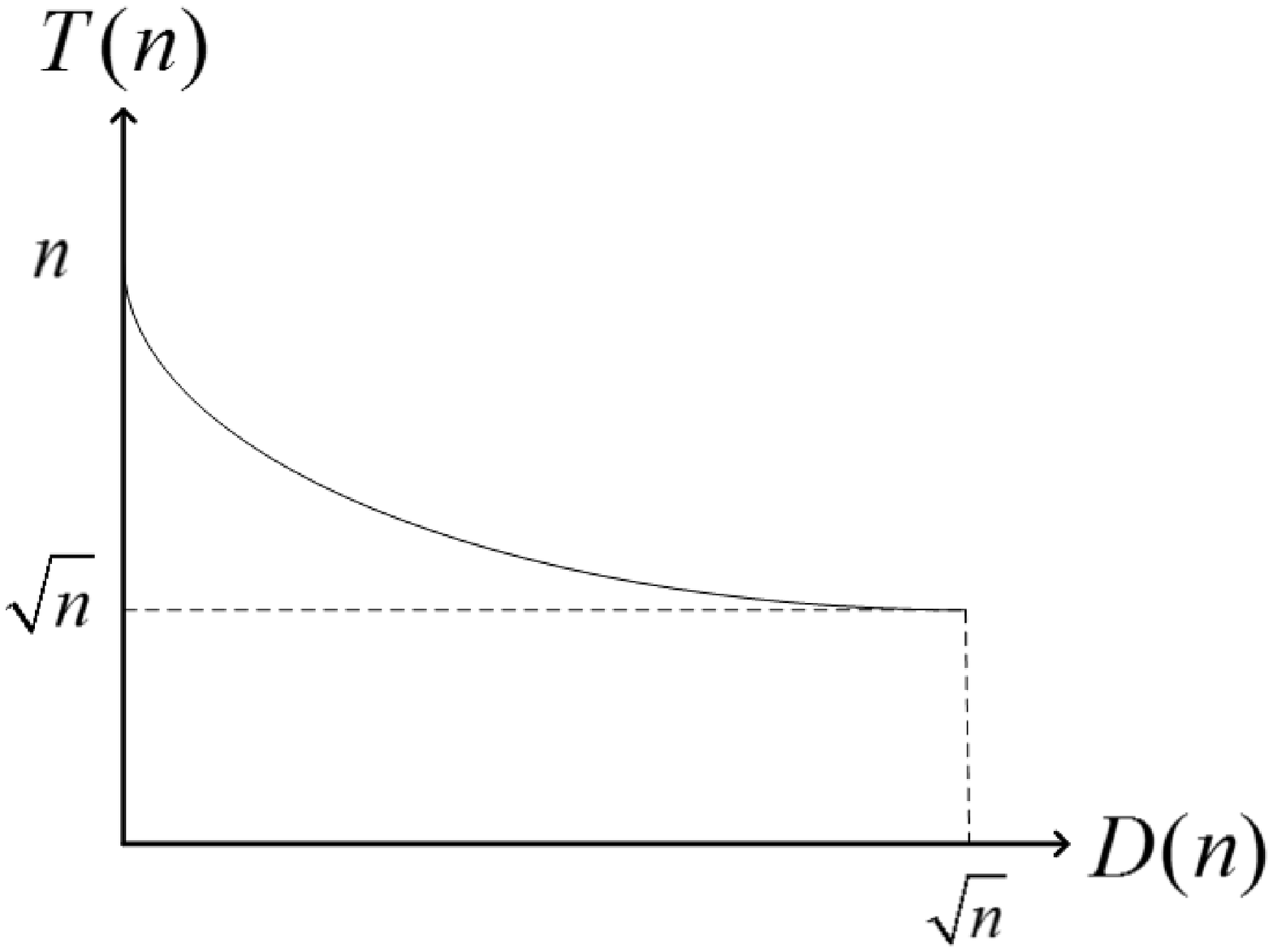}}} \caption{The trade-off between the
network delay $D(n)$ and the aggregate throughput $T(n)$ according
to the types of the transmit--receive antenna.} \label{FIG:delay}
\end{figure}

To be specific, with regard to Remark~\ref{Rem:delay}, when we
define the network delay $D(n)$ as the average number of hops per
SD pair, denoted by $\bar{h}$, a trade-off between the network
delay scaling and the aggregate throughput scaling can be
illustrated in Fig. \ref{FIG:delay}. In the omnidirectional mode,
the network throughput increases up to $\Theta(\sqrt{n})$ while
the network delay gets also
increased~\cite{ElGamalMammenPrabhakarShah:06}. On the other hand,
with the use of directional antennas, the proposed elastic routing
enables to increase the per-hop distance during the packet
transmission, resulting in the reduced network delay and the
increased number of simultaneously active SD pairs.

\subsection{Throughput Scaling in Extended Networks}

We now analyze the throughput scaling for extended networks using
the elastic routing protocol. From the fact that
Lemma~\ref{LEM:Int1} also holds for extended networks, the
following theorem establishes our second main result.

\begin{theorem}
In the extended ad hoc network of unit node density, the aggregate
throughput achieved by elastic routing is identical to the dense
network case as in Theorem \ref{THM:ELASTIC}.
\end{theorem}

\begin{proof}
In the extended network, by following the same analysis as in
(\ref{EQ:SINR11})--(\ref{EQ:R_n_lower}), we have
\begin{align}
R(n)\overset{\mbox{\footnotesize
w.h.p.}}{\geq}\left(\frac{2M(n)}{n}-\frac{1}{n}\right)\log\left(1+\frac{\frac{c_3G_m^2}{d_{\text{hop}}^\alpha}}{N_0+\frac{c_4G_m^2}{d_{\text{hop}}^\alpha}+
c_5}\right),
\end{align}
where $d_{\text{hop}}$ is given by (\ref{EQ:d_hop_extended}) and
$c_3$, $c_4$, and $c_5$ are some positive constants. Since
\begin{align}
\frac{G_m^2}{d_{\text{hop}}^\alpha} = \left\{
\begin{array}{ll} \Theta\left(\frac{1}{(\log n)^{\alpha/2}}\right) & \mbox{if
$\theta^{-1}\!\!=\!\!o\left(\left(\frac{n}{\log
n}\right)^{\alpha/4}\right)$}\\ \Omega\left(\frac{1}{(\log
n)^{\alpha/2}}\right) & \mbox{if
$\theta^{-1}\!\!=\!\!\Omega\left(\left(\frac{n}{\log
n}\right)^{\alpha/4}\right)$} \end{array} \right. \nonumber
\end{align}
from (\ref{EQ:d_hop_extended}), it follows that $R(n)=
\Omega\left(\frac{M(n)}{n}\right)$ and $T(n)=
\Omega\left(M(n)\right)$ are achievable w.h.p., which completes
the proof of the theorem.\footnote{The throughput degradation up
to a polylogarithmic term, coming from a power limitation, is
omitted since it is negligible.}
\end{proof}

As in the achievability result based on the nearest-neighbor
multihop in the omnidirectional mode~\cite{GuptaKumar:00}, note
that, when elastic routing is used in the network, the aggregate
throughput scaling for the {extended} network is the same as the
{dense} network configuration for all operating regimes with
respect to $\theta$. That is, the network configuration type does
not essentially change our scaling result as long as packet
forwarding protocols such as the nearest-neighbor multihop routing
and elastic routing are concerned.

\section{Extension to Hybrid Networks: The Use of Infrastructure} \label{SEC:DISCUSS}


\begin{figure*}[t!]
  \begin{center}
  \leavevmode \epsfxsize=0.79\textwidth   
  \leavevmode 
\epsffile{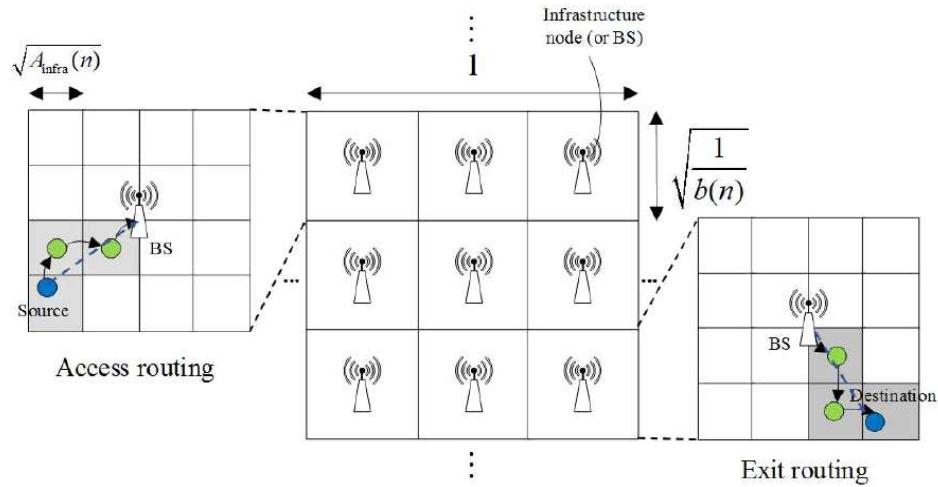} \caption{The hybrid network
with infrastructure support, where the dense network is assumed.}
\label{FIG:Infra}
  \end{center}
\end{figure*}

In this section, we consider hybrid networks by deploying
infrastructure aiding wireless nodes. Such hybrid networks,
consisting of both ad hoc nodes and infrastructure nodes, have
been introduced, and their throughput scaling laws were analyzed
in~\cite{KozatTassiulas:03,ZemlianovVeciana:05,LiuLiuTowsley:03,LiuThiranTowsley:07,ShinJeonDevroyeVuChungLeeTarokh:08,LiFang:10,JeonChung:12}.
In a hybrid network equipping directional antennas at each node,
we analyze the impact and benefits of the proposed elastic routing
in further improving the throughput scaling.

\subsection{System Model}

The whole network area is divided into $b(n)$ square cells, each
of which is covered by one BS equipped with a single directional
antenna at its center~(see Fig.~\ref{FIG:Infra}), which similarly
follows the system model
in~\cite{KozatTassiulas:03,ZemlianovVeciana:05,LiuLiuTowsley:03,LiuThiranTowsley:07}.
For analytic convenience, let us state that parameters $n$ and
$b(n)$ are related according to $b(n) = n^{\gamma}$ for $\gamma
\in [0,1)$. Moreover, as in
\cite{KozatTassiulas:03,ZemlianovVeciana:05,LiuLiuTowsley:03,LiuThiranTowsley:07,ShinJeonDevroyeVuChungLeeTarokh:08},
it is assumed that BSs are connected to each other by wired
infrastructure with infinite bandwidth (i.e., infinite capacity)
and that they are neither sources nor destinations. We assume that
each BS has an average transmit power constraint $P$ (constant)
and CSI is available at the receive side including the receive
BSs, but not at the transmit side including the transmit BSs. We
do not assume the use of any sophisticated multiuser detection
schemes at each receiver, thereby resulting in an easier
implementation.

\subsection{Infrastructure-Supported Elastic Routing Protocol}

We state our result based on the dense network model, but the
overall procedure can be straightforwardly applied to the extended
network configuration. Now, we describe a new elastic routing
scheme utilizing infrastructure, i.e., BSs. For the
infrastructure-supported elastic routing scheme, each square cell
is tessellated into several square {\em routing cells} having the
area of $A_{\text{infra}}(n)$ (which will be specified later). As
illustrated in Fig.~\ref{FIG:Infra}, the infrastructure-supported
elastic routing consists of three phases: {\em access routing},
{\em BS-to-BS transmission}, and {\em exit routing}. In the access
routing, a source node transmits its packet to the nearest BS. The
packet is then delivered to another BS that is nearest to the
destination of the source via wired link. The packet is received
at the destination from the BS in the exit routing. To avoid a
large amount of interference, different time slots are used
between not only routing schemes with and without infrastructure
support but also access and exit routings. The
infrastructure-supported elastic routing protocol is described
more specifically as follows:

\begin{itemize}
\item Divide the network into equal square cells of area $1/b(n)$,
each having one BS at the center of each cell, and again divide
each cell into smaller square routing cells of area
$A_{\text{infra}}(n)$, specified in
(\ref{EQ:RoutingRegion_Infra}).

\item In the access routing, one source in each cell transmits its
packets to the corresponding BS via the elastic routing, using one
of the nodes in each adjacent routing cell. We draw the straight
line connecting a source to its nearest BS. For convenience, we
term the line a ``source--BS line''. Then, packets for an SD pair
travel horizontally or vertically along its source--BS line by
hopping along adjacent routing cells of area $A_{\text{infra}}(n)$
until they reach the nearest BS. Similarly as in the pure ad hoc
transmission case, while travelling along a source--BS line, a
certain node in each routing cell is arbitrarily selected as a
relay forwarding the packets. As in the omnidirectional mode, at
each hop, a transmit power of $P(\log n / n)^{\alpha/2}$ is used,
and the antennas of each selected transmitter--receiver pair are
steered so that their beams cover each other.



\item The BS that completes decoding its packets transmits them to
the BS closest to the corresponding destination by wired BS-to-BS
links.

\item In the exit routing, similarly as in the access routing, the
infrastructure-supported elastic routing from a BS to the
corresponding destination is performed, where each BS uses power
$P(\log n / n)^{\alpha/2}$ that satisfies the power constraint. We
draw the straight line connecting a destination to its nearest BS.
Then, packets for an SD pair travel horizontally or vertically
along the drawn BS--destination line by hopping along adjacent
routing cells of area $A_{\text{infra}}(n)$ until they reach their
destination. At each hop, the antennas of each selected
transmitter--receiver pair are also steered so that their beams
cover each other.
\end{itemize}

When $b(n)$ BSs are deployed over the hybrid network, the maximum
distance that packets travel through air interface is limited by
$O\left(\sqrt{\frac{1}{b(n)}} \right)$ if an
infrastructure-supported routing protocol is used. Thus, similarly
as in the ad hoc network case, the area of the routing cell,
$A_{\text{infra}}(n)$, is given by
\begin{align}
A_{\text{infra}}(n) = \Theta\left(d_{\text{infra-hop}}^2\right) =
\Theta\left(\frac{1}{b(n)\bar{h}_{\text{infra}}^2}\right)
\label{EQ:RoutingRegion_Infra}
\end{align}
in the dense hybrid network, where $d_{\text{infra-hop}}$ and
$\bar{h}_{\text{infra}}$ denote the average per-hop distance and
the average number of hops, respectively, in the
infrastructure-supported elastic routing. For given
$A_{\text{infra}}(n)$, the average per-hop distance
$d_{\text{infra-hop}}$ can be determined using
(\ref{EQ:RoutingRegion_Infra}) and is shown according to the value
of $\theta$ as follows:
\begin{align}
&d_{\text{infra-hop}} = \nonumber \\
& \left\{
\begin{array}{ll}
{\Theta\left(\sqrt{\frac{\log n}{ n}}
\theta^{-\frac{2}{\alpha}}\right)}
&\textrm{if
$\theta^{-1}\!\!=\!\!o\left(\left(\frac{n}{b(n)\log n}\right)^{\alpha/4}\right)$} \\
{\Theta\!\left(\sqrt{\frac{1}{b(n)}}\right)}
&\textrm{if $\theta^{-1}\!\!=\!\!\Omega\left(\left(\frac{n}{b(n)
\log n}\right)^{\alpha/4}\right)$},
\end{array}
\right. \label{EQ:d_hop_infra}
\end{align}
where $b(n) = n^{\gamma}$ for $\gamma \in [0,1)$, which will be
verified in the next subsection.

Note that, for the given transmit power (i.e., the transmit power
$P(\log n / n)^{\alpha/2}$), per-hop distance in the
infrastructure-supported elastic routing is longer than that in
the conventional infrastructure-supported multihop of the
omnidirectional
mode~\cite{KozatTassiulas:03,ZemlianovVeciana:05,LiuLiuTowsley:03,LiuThiranTowsley:07}
for all $\theta$.

\subsection{Throughput Scaling}

In the following theorem, we establish our third main result,
which presents the aggregate throughput $T(n)$ with respect to
$\theta^{-1}$ when the proposed elastic routing protocol is used
in the dense hybrid network equipping directional antennas.

\begin{theorem} \label{THM:INFRA}
In the dense network using infrastructure, the aggregate
throughput achieved by elastic routing with and without BS support
is given by
\begin{align} \label{EQ:Tn_infra}
\lefteqn{T(n) = \Omega \left(\min\left\{
\max\left\{\sqrt{n}\theta^{-2/\alpha},b(n)\right\},n\right\}n^{-\epsilon}\right)} & \nonumber\\
&=\left\{
\begin{array}{ll}
\Omega\left(\max\{n^{1/2-\epsilon},b(n)\}\right) \\
~~~~~~~~~~~~~\textrm{if $\theta^{-1} = o \left(b(n)^{\alpha/2}
\left(\frac{\log n}{n}
\right)^{\alpha/4} \right)$} \\
\Omega\left(n^{1/2-\epsilon}\theta^{-2/\alpha}\right) \\
~~~~~~~~~~~~~ \textrm{if $\theta^{-1} = \Omega
\left(b(n)^{\alpha/2} \left(\frac{\log n}{n} \right)^{\alpha/4}
\right)$} \\~~~~~~~~~~~~~ \textrm{and $\theta^{-1} =
o\left(\left(\frac{n}{\log n}
\right)^{\alpha/4}\right)$} \\
\Omega\left(n^{1-\epsilon}\right) \\ ~~~~~~~~~~~~~ \textrm{if
$\theta^{-1}=\Omega\left(\left(\frac{n}{\log n}
\right)^{\alpha/4}\right)$}
\end{array}
\right.
\end{align}
w.h.p., where $\epsilon>0$ is an arbitrarily small constant.
\end{theorem}

\begin{proof}
We provide a brief sketch of the proof since the proof essentially
follows almost the same line as in Theorem~\ref{THM:ELASTIC}.
Under the hybrid network, the aggregate throughput $T(n)$ is
lower-bounded by
\begin{align}
T(n) \ge \max \{ T_{\text{infra}}(n), T_{\text{adhoc}}(n) \},
\nonumber
\end{align}
where $T_{\text{infra}}(n)$ and $T_{\text{adhoc}}(n)$ denote the
aggregate throughputs achieved by the infrastructure-supported and
pure ad hoc elastic routing protocols, respectively.

From the fact that $T_{\text{adhoc}}(n)$ is given by
(\ref{EQ:Tn_dense}) in Theorem \ref{THM:ELASTIC}, let us now focus
on computing $T_{\text{infra}}(n)$. We start from dealing with the
{\em access routing}. The transmission rate of each hop is
expressed as a function of the SINR value at each receiver that
the packet of a source goes through via the
infrastructure-supported elastic routing until it reaches the
corresponding BS. Let $I^{[3]}_{p,l}(s)$ and $I^{[4]}_{p,l}(s)$
denote the intra-pair and inter-pair interference powers of ${\sf
rx}_{p,l}$ (either an ad hoc node or a BS), respectively, for the
$l$th hop of the source--BS line corresponding to the $p$th SD
pair at time slot $s\in\{1,\cdots, T_S\}$.

Then, by following the same analysis as in (\ref{EQ:P_IHI}) and
(\ref{EQ:Inter_Pair_Interference}), we have
\begin{align}
I^{[3]}_{p,l}(s)& =
O\left(G_m^2 \left(\frac{\log n}{d_{\text{infra-hop}}^2 n}
\right)^{\alpha / 2} \right),
\label{EQ:P_IHI_Infra}\\
I^{[4]}_{p,l}(s) & \le P\sum_{t=1}^{\infty}
\sum_{i=1}^{8t}t^{-\alpha}
Y_{p,l,t,i}(s),\label{EQ:Inter_Pair_Interference_Infra}
\end{align}
where $Y_{p,l,t,i}(s) \in \{G_m^2, G_m G_s, G_s^2\}$ denotes the
antenna gain between ${\sf rx}_{p,l}$ and ${\sf in}_{p,l,t,i}(s)$
at time slot $s$. Under the infrastructure-supported routing
protocol, there are $b(n)$ source--BS lines that are active
simultaneously at each time slot. Thus, similarly as in the proof
of Theorem \ref{THM:ELASTIC}, the achievable per-node transmission
rate during the access routing, $R_{\text{access}}(n)$ is
lower-bounded by
\begin{align}
R_{\text{access}}(n)&\geq\left(\frac{2b(n)}{n}-\frac{1}{n}\right)\nonumber\\
&\cdot\log\left(1+\frac{c_6 G_m^2 \left(\frac{\log
n}{d_{\text{infra-hop}}^2 n} \right)^{\alpha / 2}}{  N_0+c_7 G_m^2
\left(\frac{\log n}{d_{\text{infra-hop}}^2 n} \right)^{\alpha / 2}
+ c_8}\right), \label{EQ:avgTn_Infra}
\end{align}
w.h.p., where $d_{\text{infra-hop}}$ is given by
\eqref{EQ:d_hop_infra} and $c_6$, $c_7$, and $c_8$ are some
positive constants. By substituting (\ref{EQ:d_hop_infra}) into
(\ref{EQ:avgTn_Infra}), it follows that
\begin{align}
R_{\text{access}}(n)=\Omega\left(\frac{b(n)}{n}\right) \nonumber
\end{align}
is achievable w.h.p. In a similar fashion, it can also be shown
that $R_{\text{exit}}(n)=\Omega((b(n)/n)$ is achievable w.h.p.
during the {\em exit routing}. Hence, it is shown that
$R_{\text{infra}}(n)=\Omega\left(\frac{b(n)}{n}\right)$ is
achievable w.h.p. and, as a consequence,  $T_{\text{infra}}(n)=
\Omega\left(b(n)\right)$ is achievable w.h.p. In conclusion, the
aggregate throughput $T(n)$ is given by the result in
(\ref{EQ:Tn_infra}), which completes the proof of the theorem.
\end{proof}

From Theorem \ref{THM:INFRA}, it is shown that, as the number of
BSs $b(n)$ scales slower than $n^{1/2}$, the throughput scaling
and operating regimes remain the same as those of no
infrastructure model (i.e., the ad hoc network model). However,
when $b(n)$ scales faster than $n^{1/2}$, the operating regimes
and the best achievable schemes in each regime are considerably
changed and are summarized as follows:

\begin{itemize}
\item {\bf Regime III}:
$\theta^{-1} = o \left(b(n)^{\alpha/2} \left(\frac{\log
n}{n}\right)^{\alpha/4} \right)$.

The infrastructure-supported elastic routing is used, and its
throughput scaling is given by $b(n)$. In the regime, the use of
directional antennas is not useful in terms of throughput scaling
laws.

\item {\bf Regime IV}:
$\theta^{-1} = \Omega \left(b(n)^{\alpha/2} \left(\frac{\log
n}{n}\right)^{\alpha/4} \right)$ and
$\theta^{-1}=o\left(\left(\frac{n}{\log n}
\right)^{\alpha/4}\right)$.

In Regime IV, the ad hoc elastic routing outperforms the
infrastructure-supported elastic routing. Thus, the use of
infrastructure is not helpful in this regime. The aggregate
throughput scaling
$\Omega\left(n^{1/2-\epsilon}\theta^{-2/\alpha}\right)$ is
achieved, which is the same as that in Regime I assuming no BSs.

\item {\bf Regime V}: 
$\theta^{-1}=\Omega\left(\left(\frac{n}{\log n}
\right)^{\alpha/4}\right)$

The single-hop transmission is performed in this regime where a
linear aggregate throughput scaling is achieved, which shows the
same scaling as Regime II.
\end{itemize}


\begin{figure}[t!]
  \begin{center}
  \leavevmode \epsfxsize=0.46\textwidth   
  \leavevmode 
\epsffile{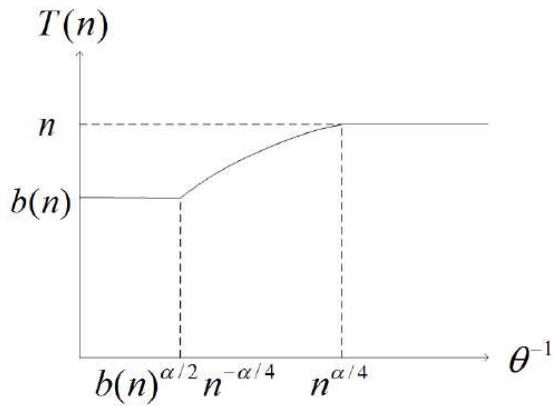} \caption{The aggregate throughput
scaling $T(n)$, achieved by elastic routing, with respect to the
inverse of the beam width, $1/\theta$, where
$b(n)=\Omega(n^{1/2})$ is assumed.} \label{FIG:Infra_Throughput}
  \end{center}
\end{figure}

In the hybrid network assuming that $b(n) = \Omega(n^{1/2})$, the
aggregate throughput scaling $T(n)$ is illustrated in Fig.
\ref{FIG:Infra_Throughput}. Note that, when $b(n) = o(n^{1/2})$,
the throughput scaling is the same as Theorem~\ref{THM:ELASTIC},
which is shown in Fig. \ref{FIG:Throughput}. When the beam width
$\theta$ of each directional antenna is not sufficiently narrow
(Regime III), the throughput scaling is given by the number of
BSs, $b(n)$, where the throughput scaling is improved with the
help of the BSs but the use of directional antennas is not helpful
in further enhancing the throughput performance. As $\theta$
becomes narrower (Regime IV), the throughput scaling is improved
with increasing $\theta^{-1}$ and thus the directional antenna
gain can be attainable, where the pure ad hoc elastic routing is
dominant in the regime. When $\theta^{-1} = \Omega\left(
\left(\frac{n}{\log n}\right)^{-\alpha/4}\right)$ (Regime V), a
linear throughput scaling is achieved, which eventually approaches
the fundamental limit of the network.

In addition, we remark that, since each source/BS can reach the
corresponding BS/destination via one hop using a transmit power of
$P b(n)^{-\alpha/2}\theta^2$, an infrastructure-supported
single-hop directional transmission in each cell can also be
performed while achieving the same throughput scaling law as the
infrastructure-supported elastic routing case.

In the extended hybrid network, a full transmit power $P$ is used
at each node (including each BS) since the network is also
power-limited. Then, one can easily see that the throughput
scaling and fundamental operating regimes are the same as those in
the dense network.

\section{Conclusion} \label{SEC:CONCL}

This paper has analyzed throughput scaling laws of ad hoc networks
equipping directional antennas, each of which has a scalable beam
width $\theta$ with respect to the number of nodes, $n$. To fully
utilize the characteristics of directional antennas, the elastic
routing protocol was proposed, where per-hop distance is increased
elastically according to $\theta$ while the average received SINR
is maintained as a constant. Under the proposed routing protocol,
fundamental operating regimes with respect to $\theta$ and the
corresponding throughput scaling laws were identified. It was
proved that the elastic routing protocol exhibits a much higher
throughput scaling result compared to the conventional multihop
scheme as $\theta^{-1}$ increases. Moreover, our result was
generalized to the hybrid network scenario with infrastructure.
The impacts and benefits of the elastic routing protocol were
comprehensively analyzed in further improving throughput scaling
laws in the hybrid network.


\end{document}